\documentclass[fleqn,11pt]{article}
\usepackage{amssymb}

\usepackage{amsmath}
\usepackage{float}
\usepackage{theorem}

\theorembodyfont{\upshape}
\newtheorem{theorem}{Theorem}

\newtheorem{claim}[theorem]{Claim}

\newtheorem{definition}[theorem]{Definition}
\newtheorem{example}[theorem]{Example}

\newtheorem{notation}[theorem]{Notation}

\newtheorem{remark}[theorem]{Remark}

\newenvironment{proof}[1][Proof]{\textbf{#1.} }{\ \rule{0.5em}{0.5em}}
\input{tcilatex}
\advance\topmargin by -0.75truein
\advance\textheight by 1.25truein
\begin{document}

\begin{center}
{\LARGE Limits to the scope of applicability of \textit{extended formulations%
\\[0pt]
}for LP models of combinatorial optimization problems}{\Large \medskip
\medskip }

Moustapha Diaby

OPIM Department; University of Connecticut; Storrs, CT 06268\\[0pt]
moustapha.diaby@business.uconn.edu\bigskip

Mark H. Karwan

Department of Industrial and Systems Engineering; SUNY at Buffalo; Amherst,
NY 14260\\[0pt]
mkarwan@buffalo.edu\bigskip\ \ 
\end{center}

\textsl{Abstract}{\small : The purpose of this paper is to bring to
attention and to make a contribution to the issue of defining/clarifying the
scope of applicability of \textit{extended formulations} (EF's) theory.
Specifically, we show that EF theory is not valid for relating the sizes of
descriptions of polytopes when the sets of the descriptive variables for
those polytopes are disjoint, and that new definitions of the notion of
``projection'' upon which some of the recent \textit{extended formulations}
works (such as Kaibel (2011), Fiorini\textit{\ et al.} (2011, 2012), Kaibel
and Walter (2013), and Kaibel and Weltge (2013), for example) have been
based can cause those works to over-reach in their conclusions.\bigskip }

\textsl{Keywords:}\textbf{\ }\textit{Extended Formulations}{\small ; Linear
Programming; Combinatorial Optimization; Computational Complexity.}$\medskip 
$

\section{Introduction\label{Introduction_Section}}

There has been a renewed interest and excellent work in \textit{extended
formulations} (EF's) over the past few years (see Conforti \textit{et al}.
(2010; 2013), and Vanderbeck and Wolsey (2010), for example). Since the
seminal paper (Yannakakis (1991)), EF theory has been the single-most
important paradigm for deciding the validity of proposed LP models for
NP-Complete problems. However, the issue of its scope of applicability has
been a largely overlooked issue, leading to the possibility of over-reaching
claims (implied or explicitly-stated). The purpose of this paper is to make
a contribution towards addressing the issue of delineating the scope of
applicability of EF theory. Specifically, we will show that EF theory is not
valid for relating the sizes of descriptions of polytopes when the sets of
the descriptive variables for those polytopes are disjoint, and that new
definitions of the notion of ``projection'' upon which some of the recent 
\textit{extended formulations} works (such as Kaibel (2011), Fiorini\ et al.
(2011, 2012), Kaibel and Walter (2013), and Kaibel and Weltge (2013), for
example) have been based can cause those works to over-reach in their
conclusions when the sets of the descriptive variables for the polytopes
being related are disjoint or can be made so after redundant variables and
constraints (with respect to the optimization problem at hand) are removed.

It should be noted that the intent of the paper is not to claim the
correctness or incorrectness of any particular model that may have been
developed in trying to address the ``$P=NP$'' question. Our aim is,
strictly, to bring attention to limits to the scope within which EF theory
is applicable when attempting to derive bounds on the sizes of linear
programming models for combinatorial optimzation problems. In other words,
the developments in the paper are not about deciding the
correctness/incorrectness of any given LP model, but only about the issue of
when such a decision (of correctness/incorrectness) is beyond the scope of
EF theory.

One of the most fundamental assumptions in EF theory is that the addition of
redundant variables and constraints to a given model of an optimization
problem at hand does not change the EF relationships for that model. The key
point of this paper is to show the flaw in this assumption, which is that it
leads to ambiguity and degeneracy/loss of meaningfulness of the notion of EF
when the sets of the descriptive variables for the polytopes involved are
disjoint. We show that if redundant variables and constraints can be
arbitrarily added to the description of a given polytope for the the purpose
of establishing EF relationships, then every given mathematical programming
model would be an EF of every other mathematical programming model, provided
their sets of descriptive variables are disjoint, which would clearly mean a
loss of meaningfulness of the notion (of EF).

Our developments in this paper were initially motivated by our realization
that the ``new'' definition of EFs (Definition \ref{Extended_Polytope_Dfn2})
which was first proposed in Kaibel (2011) and then subsequently used in
Fiorini et al. (2011) becomes inconsistent with other and
previous/``standard'' definitions of EFs (see Definition \ref%
{Extended_Polytope_Dfn}) when the sets of the descriptive variables of the
polytopes being related are disjoint. Comments we received in private
communications and also in anonymous reviews on an earlier version of this
paper were that our finding of inconsistency was ``obvious,'' and that all
of our developments were ``obvious'' because of that. Hence, we believe it
may be useful to recall at this point that the Fiorini et al. (2011) work
(which is based on the ``new'' definition in question) has been
highly-recognized, having received numerous awards. Also, more generally,
and perhaps more deeply, Martin's formulation of the Minimum Spanning Tree
Problem (MSTP; see section \ref{Min_Span_Tree_SubSection} of this paper) is
cited in almost every EF paper in the current literature as a ``normal'' EF
of Edmonds' formulation of the MSTP (although with the acknowledgement that
it ``escapes'' the results for EF of NP-Complete problems somehow). We argue
that these facts highlight the need to bring our notions in this paper to
the attention of the Optimization communities in general, and of the EF
communities in particular.

The plan of the paper is as follows. First, we will review the background
definitions in section \ref{Background_Section}. Our main result (i.e., the
non-validity/non-applicability of EF theory when the sets of descriptive
variables are disjoint) is developed in section \ref%
{Non-applicability_Section}. In section \ref{Illustration_Section}, we
illustrate the discussions of section \ref{Non-applicability_Section} using
the Fiorini et al. (2011; 2012) developments, as well as Martin's (1991) LP
formulation of the Minimum Spanning Tree Problem (MSTP). In section \ref%
{Auxiliary_Model_Section}, we provide insights into the (correct)
meaning/implication of the existence of a linear map between solutions of
models, with respect to the task of solving an optimzation problem, when the
set of descriptive variables in the models are disjoint. Finally, we offer
some concluding remarks in section \ref{Conclusions_Section}.\medskip

The general notation we will use is as follows.

\begin{notation}
\label{TSP_Gen'l_Notations}\ \ 
\end{notation}

\begin{enumerate}
\item $\mathbb{R}:$ \ Set of real numbers; \ 

\item $\mathbb{R}_{\mathbb{\nless }}:$ \ Set of non-negative real numbers;

\item $\mathbb{N}:$ \ Set of natural numbers;

\item $\mathbb{N}_{\mathbb{+}}:$ \ Set of positive natural numbers;

\item $``\mathbf{0}":$ \ Column vector that has every entry equal to $0$;

\item $``\mathbf{1}":$ \ Column vector that has every entry equal to $1$;

\item $(\cdot )^{T}:$ \ Transpose of $(\cdot )$;

\item $Conv(\cdot ):$ \ Convex hull of $(\cdot )$.\medskip
\end{enumerate}

\section{Background definitions\label{Background_Section}}

For the purpose of making the paper as self-contained as possible, we review
the basic definitions of \textit{extended formulations} in this section.

\begin{definition}[``Standard EF Definition'' (Yannakakis\ (1991); Conforti 
\textit{et al}. (2010; 2013))]
\label{Extended_Polytope_Dfn}An \textit{extended formulation} for a polytope 
$X$ $\subseteq $ $\mathbb{R}^{p}$ is a polyhedron $U$ $=$ $\{(x,w)$ $\in $ $%
\mathbb{R}^{p+q}$ $:$ $Gx$ $+$ $Hw$ $\leq $ $g\}$ the projection, $\varphi
_{x}(U)$ $:=$ $\{x\in \mathbb{R}^{p}:$ $(\exists w\in \mathbb{R}^{q}:$ $%
(x,w) $ $\in $ $U)\},$ of which onto $x$-space is equal to $X$ (where $G$ $%
\in \mathbb{R}^{m\times p},$ $H\in \mathbb{R}^{m\times q},$ and $g\in 
\mathbb{R}^{m}$).
\end{definition}

\begin{definition}[``Alternate EF Definition \#1\textit{''} (Kaibel (2011);
Fiorini \textit{et al}. (2011; 2012))]
\label{Extended_Polytope_Dfn2}A polyhedron $U$ $=$ $\{(x,w)$ $\in $ $\mathbb{%
R}^{p+q}$ $:$ $Gx$ $+$ $Hw$ $\leq $ $g\}$ is an \textit{extended formulation}
of a polytope $X$ $\subseteq $ $\mathbb{R}^{p}$ if there exists a linear map 
$\pi $ $:$ $\mathbb{R}^{p+q}$ $\longrightarrow $ $\mathbb{R}^{p}$ such that $%
X$ is the image of $Q$ under $\pi $ (i.e., $X=\pi (Q)$; where $G\in \mathbb{R%
}^{m\times p}$, $H\in \mathbb{R}^{m\times q},$ and $g\in \mathbb{R}^{m}$).
(Kaibel (2011), Kaibel and Walter (2013), and Kaibel and Weltge (2013) refer
to $\pi $ as a ``projection.'')
\end{definition}

\begin{definition}[``Alternate EF Definition \#2'' (Fiorini \textit{et al}.
(2012))]
\label{Extended_Polytope_Dfn3}An \textit{extended formulation} of a polytope 
$X$ $\subseteq $ $\mathbb{R}^{p}$ is a linear system $U$ $=$ $\{(x,w)$ $\in $
$\mathbb{R}^{p+q}$ $:$ $Gx$ $+$ $Hw$ $\leq $ $g\}$ such that $x\in X$ if and
only if there exists $w\in \mathbb{R}^{q}$ such that $(x,w)\in U.$ (In other
words, $U$ is an EF of $X$ if $(x\in X\Longleftrightarrow (\exists $ $w\in 
\mathbb{R}^{q}:(x,w)\in U))$) (where $G$ $\in \mathbb{R}^{m\times p},$ $H\in 
\mathbb{R}^{m\times q},$ and $g\in \mathbb{R}^{m}$)$.\medskip $
\end{definition}

The purpose of this paper is to point out that the scope of applicability of
EF work based on the above definitions is limited to cases in which $U$
cannot be equivalently reformulated (with respect to the task of optimizing
linear functions) in terms of the $w$ variables only. For simplicity of
exposition, without loss of generality, we will say that $G=\mathbf{0}$ in
the above definitions if there exists a description of $U$ which is in terms
of the $w$-variables only and has the same or smaller complexity order of
size. Or, equivalently, without loss of generality, we will say that $G\neq 
\mathbf{0}$ in the above definitions iff the $x$- and $w$-variables are
required in every valid inequality description of $U$ which has the same or
smaller complexity order of size as the description at hand.

In particular, if every constraint of $U$ which involves the $x$-variables
is redundant in the description of $U$, then clearly, every one of those
constraints as well as the $x$-variables themselves can be dropped (without
loss, with respect to the task of optimizing linear functions) from the
description of $U$, with the result that $U$ would be stated in terms of the 
$w$-variables only. Also, in some cases (all of) the contraints involving
the $x$-variables may become redundant only after other constraints in the
description of $U$ are re-written and/or new constraints are added (as
exemplified by the case of the minimum spanning tree problem (MSTP) in
section \ref{Min_Span_Tree_SubSection} of this paper). If either of these
two cases is applicable, we will say that $G=\mathbf{0}$ in the above
definitions. Otherwise, we will say that $G\neq \mathbf{0.}$

\begin{remark}
\label{EF_Dfns_Observations_Rmk}The following observations are in order with
respect to Definitions \ref{Extended_Polytope_Dfn}, \ref%
{Extended_Polytope_Dfn2}, and \ref{Extended_Polytope_Dfn3}:

\begin{enumerate}
\item The statement of $U$ in terms of inequality constraints only does not
cause any loss of generality, since each equality constraint can be replaced
by a pair of inequality constraints.

\item The system of linear equations which specify $\pi $ in Definition \ref%
{Extended_Polytope_Dfn2} must be \textit{valid} constraints for $X$ and $U$.
Hence, $X$ and $U$ can be respectively \textit{extended} by adding those
constraints to them, when trying to relate $X$ and $U$ using Definition \ref%
{Extended_Polytope_Dfn2}. In that sense, Definition \ref%
{Extended_Polytope_Dfn2} ``extends'' Definitions \ref{Extended_Polytope_Dfn}
and \ref{Extended_Polytope_Dfn3}.

\item All three definitions are equivalent when $G\neq \mathbf{0}$. However,
this is not true when $G=\mathbf{0,}$ as we will show in section \ref%
{Non-applicability_Section} of this paper.

\item In the remainder of this paper, we will use the term ``polytope'' to
refer to the polytope induced by a set of linear inequality constraints or
the set of constraints itself, if this is convenient and does not cause
ambiguity.
\end{enumerate}

\noindent $\square $
\end{remark}

\section{Non-applicability and degeneracy\newline
conditions for \textit{extended formulations}\label%
{Non-applicability_Section}}

Our main result will now be developed.

\begin{theorem}
\label{Thm1}EF developments are not valid for relating the inequality
descriptions of $U$ and $X$ in Definitions \ref{Extended_Polytope_Dfn}-\ref%
{Extended_Polytope_Dfn3} when $G=\mathbf{0}$ in those definitions.\medskip
\end{theorem}

\begin{proof}
The proof will be in three parts. In Part 1, we will show that when $G=%
\mathbf{0}$, $U$ cannot be an EF of $X$ according to Definition \ref%
{Extended_Polytope_Dfn}. In Part 2, we will show that when $G=\mathbf{0}$, $%
U $ cannot be an EF of $X$ according to Definition \ref%
{Extended_Polytope_Dfn3}. In Part 3, we will show that when $G=\mathbf{0,}$
the EF notion under Definition \ref{Extended_Polytope_Dfn2} results in the
condition that every given polytope is an \textit{extended formulation} of
every other given polytope, provided their sets of descriptive variables are
disjoint, which means that the (EF) notion becomes degenerate/meaningless.
In the discussion, we will only consider the case in which $U\neq
\varnothing $ and $X\neq \varnothing ,$ since the theorem is trivial when $%
U=\varnothing $ or $X=\varnothing $.\medskip

\begin{description}
\item $(i)$ Consider Definition \ref{Extended_Polytope_Dfn}. Assume $G=%
\mathbf{0.}$ Then, we have: 
\begin{eqnarray*}
\varphi _{x}(U) &=&\{x\in \mathbb{R}^{p}:(\exists w\in \mathbb{R}%
^{q}:(x,w)\in U)\} \\
&=&\{x\in \mathbb{R}^{p}:(\exists w\in \mathbb{R}^{q}:Hw\leq g\} \\
&=&\mathbb{R}^{p} \\
&\neq &X\text{ \ (since }X\text{ is a polytope and thus bounded, whereas }%
\mathbb{R}^{p}\text{ is unbounded).}
\end{eqnarray*}%
Hence when $G=\mathbf{0}$, $U$ cannot be an EF of $X$ according to
Definition \ref{Extended_Polytope_Dfn}.

\item $(ii)$ Consider Definition \ref{Extended_Polytope_Dfn3}. Assume $G=%
\mathbf{0.}$ Then, we have:%
\begin{equation}
(\exists w\in \mathbb{R}^{q}:\mathbf{0}x+Hw\leq g)\Longleftrightarrow
(\exists w\in \mathbb{R}^{q}:Hw\leq g)\Longrightarrow \left( \forall x\in 
\mathbb{R}^{p},(x,w)\in U)\right) .  \label{ProofEqn(1)}
\end{equation}%
Which implies: 
\begin{equation}
(\exists w\in \mathbb{R}^{q}:\mathbf{0}x+Hw\leq g)\nRightarrow x\in X\text{
\ (since }X\neq \mathbb{R}^{p}\text{)}.\text{ }  \label{ProofEqn(1.5)}
\end{equation}%
From (\ref{ProofEqn(1.5)}), the ``if and only if'' stipulation of Definition %
\ref{Extended_Polytope_Dfn3} cannot hold in general. Hence, when $G=\mathbf{0%
}$, $U$ cannot be an EF of $X$ according to Definition \ref%
{Extended_Polytope_Dfn3}.$\medskip $

\item $(iii)$ Now, consider Definition \ref{Extended_Polytope_Dfn2}. We will
show that the EF notion under this definition becomes degenerate/meaningless
when $G=\mathbf{0}$. The reasons for this are that a polytope can also be
stated in terms of its extreme points (see Rockafellar (1997, pp. 153-172),
among others), and that a linear map (as stipulated in the definition) could
be inferred from this statement without reference to an inequality
description of the polytope. The proof consists of a counter-example to the
sufficiency of the existence of a linear map, as stipulated in the
definition, for implying EF relationships, as stipulated in the definition.
Note that if $G=\mathbf{0}$ in Definition \ref{Extended_Polytope_Dfn2}, then
the linear inequality description of $U$ involves the $w-$variables only.

\qquad For the sake of simplicity (but without loss of generality), let $%
\overline{U}\subset $ $\mathbb{R}^{5}$ be described in the $w$-variables
only as%
\begin{equation}
\overline{U}=\{w\in \mathbb{R}_{\nless
}^{5}:w_{1}+w_{2}=5;w_{1}-w_{2}=1;w_{3}+w_{4}+w_{5}=0\}.  \label{ProofEqn(2)}
\end{equation}%
Then, the vertex-description of $\overline{U}$ is \medskip

$\qquad \overline{U}=\{w\in \mathbb{R}_{\nless }^{5}:w\in Conv\left( \left\{
(3,2,0,0,0)^{T}\right\} \right) \}.\medskip $

\qquad Now, let $X$ $\subset $ $\mathbb{R}_{\nless }^{3}$ be specified by
its vertex-description as 
\begin{equation*}
X=\{x\in \mathbb{R}_{\nless }^{3}:x\in Conv(\left\{ (2,1,5)^{T}\right\} )\}.
\end{equation*}%
(In other words, $X$ consists of the point in $\mathbb{R}^{3}$, $(2,1,5)^{T}$%
.)\medskip

\qquad Then, the following are true:\medskip

\begin{description}
\item $(iii.1)$%
\begin{align}
& ((x\in X)\text{ \ and \ }(w\in \overline{U}))\Longrightarrow x=Aw,  \notag
\\
& \text{where, among other possibilities, }A=\left[ 
\begin{tabular}{rrrrr}
$-1$ & $2.5$ & $2$ & $3$ & $4$ \\ 
$1$ & $-1$ & $5$ & $6$ & $7$ \\ 
$-1$ & $4$ & $8$ & $9$ & $10$%
\end{tabular}%
\right] \text{.}  \label{ProofEqn(3)}
\end{align}%
Hence, under Definition \ref{Extended_Polytope_Dfn2}, $\overline{U}$ is an
EF of every one of the infinitely-many possible inequality descriptions of $%
X $ (since $x=Aw$ in the above is a linear map between $\overline{U}$ and $X$%
).\medskip

\item $(iii.2)$ Similarly, 
\begin{align}
& ((x\in X)\text{ \ and \ }(w\in \overline{U}))\Longrightarrow w=Bx,  \notag
\\
& \text{where, among other possibilities, }B=\left[ 
\begin{tabular}{rrr}
$-1$ & $1$ & $1$ \\ 
$1$ & $-1$ & $0$ \\ 
$3$ & $1$ & $-2$ \\ 
$2$ & $-11$ & $1$ \\ 
$-10$ & $30$ & $0$%
\end{tabular}%
\right] \text{.}  \label{ProofEqn(4)}
\end{align}%
\newline
\qquad Hence, under Definition \ref{Extended_Polytope_Dfn2}, every one of
the infinitely-many possible inequality descriptions of $X$ is an EF of the
inequality descriptions of $\overline{U}$ as stated in (\ref{ProofEqn(2)})
(and, in fact, of everyone of the infinitely-many possible inequality
descriptions of $\overline{U}$).\medskip

\item $(iii.3)$ Clearly, the EF relations based on (\ref{ProofEqn(3)}) and (%
\ref{ProofEqn(4)}) above are degenerate/meaningless, since no meaningful
inferences can be made from them in attempting to compare inequality
descriptions of $\overline{U}$ and $X.$
\end{description}
\end{description}
\end{proof}

\ \ \ 

A fundamental notion in \textit{extended formulations} theory is that the
addition of redundant variables and constraints to the inequality
description of a polytope does not change the EF relationships for that
polytope. We use this fact to generalize the degeneracy/loss of
meaningfulness which arises out of Definition \ref{Extended_Polytope_Dfn2}
when $G=\mathbf{0}$ to Definitions \ref{Extended_Polytope_Dfn} and \ref%
{Extended_Polytope_Dfn3}, as follows.

\begin{theorem}
\label{Fifth_Rmk}Provided redundant constraints and variables can be
arbitrarily added to the descriptions of polytopes for the purpose of
establishing EF relationships under Definitions \ref{Extended_Polytope_Dfn}-%
\ref{Extended_Polytope_Dfn3}, the descriptions of any two given non-empty
polytopes expressed in disjoint variable spaces can be respectively \textit{%
augmented} into being \textit{extended formulations} of each other.

In other words, let $x^{1}\in \mathbb{R}^{n_{1}}$ ($n_{1}\in \mathbb{N}_{+}$%
) and $x^{2}\in \mathbb{R}^{n_{2}}$ ($n_{2}\in \mathbb{N}_{+}$) be vectors
of variables with no components in common. Then, provided redundant
constraints and variables can be arbitrarily added to the descriptions of
polytopes for the purpose of establishing EF relationships, the
inequality-description of every non-empty polytope in $x^{1}$ can be \textit{%
augmented} into an EF of the inequality-description of every other non-empty
polytope in $x^{2}$, and vice versa.
\end{theorem}

\begin{proof}
\ The proof is essentially by construction. \medskip

\noindent Let $P_{1}$ and $P_{2}$ be polytopes specified as:%
\begin{align*}
& P_{1}=\{x^{1}\in \mathbb{R}^{n_{1}}:A_{1}x^{1}\leq a_{1}\}\neq \varnothing 
\text{ \ }(\text{where }A_{1}\in \mathbb{R}^{p_{1}\times n_{1}}\text{, and }%
a_{1}\in \mathbb{R}^{p_{1}}); \\[0.09in]
& P_{2}=\{x^{2}\in \mathbb{R}^{n_{2}}:A_{2}x^{2}\leq a_{2}\}\neq \varnothing 
\text{ \ (where }A_{2}\in \mathbb{R}^{p_{2}\times n_{2}}\text{, and }%
a_{2}\in \mathbb{R}^{p_{2}}).
\end{align*}

Clearly, $\forall (x^{1},$ $x^{2})\in P_{1}\times P_{2},$ $\forall q\in 
\mathbb{N}_{\mathbb{+}}$, $\forall B_{1}\in \mathbb{R}^{q}{}^{\times n_{1}},$
$\forall B_{2}\in \mathbb{R}^{q\times n_{2}},$ there exists $u\in \mathbb{R}%
_{\nless }^{q}$ such that the constraints%
\begin{equation}
B_{1}x^{1}+B_{2}x^{2}-u\leq 0  \label{EF_Thm(a)}
\end{equation}%
are \textit{valid} for $P_{1}$ and $P_{2}$, respectively (i.e., they are 
\textit{redundant} for $P_{1}$ and $P_{2},$ respectively).\medskip

Now, consider :%
\begin{align}
W:=& \left\{ (x^{1},x^{2},u)\in \mathbb{R}^{n_{1}}\times \mathbb{R}%
^{n_{2}}\times \mathbb{R}_{\nless }^{q}:\right.  \notag \\[0.06in]
& C_{1}A_{1}x^{1}\leq C_{1}a_{1};\text{ }  \label{EF_Thm(b)} \\[0.06in]
& B_{1}x^{2}+B_{2}x^{1}-u\leq 0;  \label{EF_Thm(c)} \\[0.06in]
& \left. C_{2}A_{2}x^{2}\leq C_{2}a_{2}\right\}  \label{EF_Thm(d)}
\end{align}%
(where: $C_{1}\in $ $\mathbb{R}^{p_{1}\times }{}^{p_{1}}$ and $C_{2}$ $\in 
\mathbb{R}^{p_{2}\times }{}^{p_{2}}$ are diagonal matrices with positive
diagonal entries).\medskip

Clearly, $W$ \textit{augments} $P_{1}$ and\textit{\ }$P_{2}$ respectively.
Hence:%
\begin{equation}
W\text{ is equivalent to }P_{1},\text{ and}  \label{EF_Thm(g)}
\end{equation}%
\begin{equation}
W\text{ is equivalent to }P_{2}\text{ .}  \label{EF_Thm(h)}
\end{equation}

Also clearly, we have:%
\begin{equation}
\varphi _{x^{1}}(W)=P_{1}\text{ \ (since }P_{2}\neq \varnothing ,\text{ and
((\ref{EF_Thm(c)}) and (\ref{EF_Thm(d)}) are redundant for }P_{1})),\text{\
\ and }  \label{EF_Thm(e)}
\end{equation}%
\begin{equation}
\varphi _{x^{2}}(W)=P_{2}\text{ \ (since }P_{1}\neq \varnothing ,\text{ and
((\ref{EF_Thm(b)}) and (\ref{EF_Thm(c)}) are redundant for }P_{2})\text{)}.
\label{EF_Thm(f)}
\end{equation}

It follows from the combination of (\ref{EF_Thm(g)}) and (\ref{EF_Thm(f)})
that $P_{1}$ is an \textit{extended formulation} of $P_{2}.\medskip $

It follows from the combination of (\ref{EF_Thm(h)}) and (\ref{EF_Thm(e)})
that $P_{2}$ is an \textit{extended formulation} of $P_{1}.$\ \ \medskip\ 
\end{proof}

\begin{example}
\label{Extended_Formulation_Example}\ \ \medskip \newline
Let 
\begin{align*}
& P_{1}=\{x\in \mathbb{R}_{\nless }^{2}:2x_{1}+x_{2}\leq 6\}; \\[0.06in]
& P_{2}=\{w\in \mathbb{R}_{\nless }^{3}:18w_{1}-w_{2}\leq 23;\text{ }%
59w_{1}+w_{3}\leq 84\}.
\end{align*}%
For arbitrary matrices $B_{1},$ $B_{2}$, $C_{1},$ and $C_{2}$ (of
appropriate dimensions, respectively)$;$ say $B_{1}=\left[ 
\begin{array}{cc}
-1 & 2 \\ 
3 & -4%
\end{array}%
\right] ,$ $B_{2}=\left[ 
\begin{array}{ccc}
5 & -6 & 7 \\ 
-10 & 9 & -8%
\end{array}%
\right] ,$ $C_{1}=\left[ 7\right] ,$ and $C_{2}=\left[ 
\begin{array}{cc}
2 & 0 \\ 
0 & 0.5%
\end{array}%
\right] ;$ $P_{1}$ and $P_{2}$ can be \textit{augmented} into \textit{%
extended formulations} of each other using $u\in \mathbb{R}_{\nless }^{2}$
and $W$: 
\begin{align*}
W=& \left\{ (x,w,u)\in \mathbb{R}_{\nless }^{2+3+2}:\text{ \ }\left[ 7\right]
\left[ 
\begin{array}{cc}
2 & 1%
\end{array}%
\right] \left[ 
\begin{array}{c}
x_{1} \\ 
x_{2}%
\end{array}%
\right] \leq 42\right. ; \\
& \left[ 
\begin{array}{cc}
-1 & 2 \\ 
3 & -4%
\end{array}%
\right] \left[ 
\begin{array}{c}
x_{1} \\ 
x_{2}%
\end{array}%
\right] +\left[ 
\begin{array}{ccc}
5 & -6 & 7 \\ 
-10 & 9 & -8%
\end{array}%
\right] \left[ 
\begin{array}{c}
w_{1} \\ 
w_{2} \\ 
w_{3}%
\end{array}%
\right] -\left[ 
\begin{array}{c}
u_{1} \\ 
u_{2}%
\end{array}%
\right] \leq \left[ 
\begin{array}{c}
0 \\ 
0%
\end{array}%
\right] ; \\
& \left. \left[ 
\begin{array}{cc}
2 & 0 \\ 
0 & 0.5%
\end{array}%
\right] \left[ 
\begin{array}{ccc}
18 & -1 & 0 \\ 
59 & 0 & 1%
\end{array}%
\right] \left[ 
\begin{array}{c}
w_{1} \\ 
w_{2} \\ 
w_{3}%
\end{array}%
\right] \leq \left[ 
\begin{array}{c}
46 \\ 
42%
\end{array}%
\right] \right\} .\text{ \ }
\end{align*}%
$\square $
\end{example}

\section{Illustrations using some existing models\label{Illustration_Section}%
}

\subsection{Application to the Fiorini \textit{et al}. (2011; 2012)
developments}

Fiorini \textit{et al}. (2012) is a re-organized and extended version of
Fiorini \textit{et al}. (2011). The key extension is the addition of another
alternate defnition of \textit{extended formulations} (page 96 of Fiorini 
\textit{et al}. (2012)) which is recalled in this paper as Definition \ref%
{Extended_Polytope_Dfn3}. This new alternate definition is then used to
re-arrange ``section 5'' of Fiorini \textit{et al}. (2011) into ``section
2'' and ``section 3'' of Fiorini \textit{et al}. (2012). Hence, the
developments in ``section 5'' of Fiorini \textit{et al}. (2011) which
depended on ``Theorem 4'' of that paper, are ``stand-alones'' (as ``section
3'') in Fiorini \textit{et al}. (2012), and ``Theorem 4'' in Fiorini \textit{%
et al}. (2011) is relabeled as ``Theorem 13'' in Fiorini \textit{et al}.
(2012).

\begin{claim}
\label{FioriniClaim1}The developments in Fiorini \textit{et al}. (2011) are
not valid for relating the inequality descriptions of $U$ and $X$ in
Definitions \ref{Extended_Polytope_Dfn}-\ref{Extended_Polytope_Dfn3} when $G=%
\mathbf{0}$.
\end{claim}

\begin{proof}
Using the terminology and notation of Fiorini \textit{et al.} (2011), the
main results of section 2 of Fiorini \textit{et al.} (2011) are developed in
terms of $Q:=\{(x,y)\in \mathbb{R}^{d+k}$ $\left| \text{ }Ex+Fy=g,\text{ }%
y\in C\right. \}$ and $P:=\{x\in \mathbb{R}^{d}$ $\left| \text{ }Ax\leq
b\right. \},$ with $Q$ (in Fiorini \textit{et al. }(2011)) corresponding to $%
U$ in Definitions \ref{Extended_Polytope_Dfn}-\ref{Extended_Polytope_Dfn3}$,$
and $P$ (in Fiorini \textit{et al. }(2011)) corresponding to $X$ in
Definitions \ref{Extended_Polytope_Dfn}-\ref{Extended_Polytope_Dfn3}. Hence, 
$G=\mathbf{0}$ in Definitions \ref{Extended_Polytope_Dfn}-\ref%
{Extended_Polytope_Dfn3} corresponds to $E=\mathbf{0}$ in Fiorini \textit{et
al}. (2011). Hence, firstly, assume $E=\mathbf{0}$ in the expression of $Q$
(i.e., $Q:=\{(x,y)\in \mathbb{R}^{d+k}$ $\left| \text{ }\mathbf{0}x+Fy=g,%
\text{ }y\in C\right. \}).$ Then, secondly, consider Theorem 4 of Fiorini 
\textit{et al.} (2011) (which is pivotal in that work). We have the
following:

\begin{description}
\item $(i)$ If $A\neq \mathbf{0}$ in the expression of $P\mathbf{,}$ then
the proof of the theorem is invalid since that proof requires setting ``$%
E:=A $'' (see Fiorini \textit{et al.} (2011, p. 7));

\item $(ii)$ If $A=\mathbf{0,}$ then $P:=\{x\in \mathbb{R}^{d}$ $\left| 
\text{ }\mathbf{0}x\leq b\right. \}.$ This implies that either $P=\mathbb{R}%
^{d}$ (if $b\geq \mathbf{0}$) or $P=\varnothing $ (if $b\ngeq \mathbf{0}$).
Hence, $P$ would be either unbounded or empty. Hence, there could not exist
a non-empty polytope, $Conv(V),$ such that $P=Conv(V)$ (see Fiorini \textit{%
et al}. (2011, 16-17), among others). Hence, the conditions in the statement
of Theorem 4 of Fiorini \textit{et al.} (2011) would be \textit{ill-defined/}%
impossible.
\end{description}

\noindent Hence, the developments in Fiorini \textit{et al.} (2011) are not
valid for relating $U$ and $X$ in Definitions \ref{Extended_Polytope_Dfn}-%
\ref{Extended_Polytope_Dfn3} when $G=\mathbf{0}$ in those definitions. \ \ \ 
\end{proof}

\begin{claim}
\label{FioriniClaim2}The developments in Fiorini \textit{et al}. (2012) are
not valid for relating the inequality descriptions of $U$ and $X$ in
Definitions \ref{Extended_Polytope_Dfn}-\ref{Extended_Polytope_Dfn3} when $G=%
\mathbf{0}$.
\end{claim}

\begin{proof}
First, note that ``Theorem 13'' of Fiorini \textit{et al}. (2012, p. 101) is
the same as ``Theorem 4'' of Fiorini \textit{et al}. (2011). Hence, the
proof of Claim \ref{FioriniClaim1} above is applicable to ``Theorem 13'' of
Fiorini \textit{et al}. (2012). Hence, the parts of the developments in
Fiorini \textit{et al}. (2012) that hinge on this result (namely, from
``section 4'' onward in Fiorini et al. (2012)) are not valid for relating $U$
and $X$ in Definitions \ref{Extended_Polytope_Dfn}-\ref%
{Extended_Polytope_Dfn3} when $G=\mathbf{0}$.

Now consider ``Theorem 3'' of Fiorini \textit{et al}. (2012) (section 3,
page 99). The proof of that theorem hinges on the statement that (using the
terminology and notation of Fiorini \textit{et al.} (2012) which is similar
to that in Fiorini et al. (2011)):%
\begin{equation}
Ax\leq b\Longleftrightarrow \exists y:E^{\leq }x+F^{\leq }y\leq g^{\leq },%
\text{ }E^{=}x+F^{=}y\leq g^{=}.  \label{F2012(a)}
\end{equation}

Note that $G=\mathbf{0}$ in Definitions \ref{Extended_Polytope_Dfn}-\ref%
{Extended_Polytope_Dfn3} would correspond to $E^{\leq }=E^{=}=\mathbf{0}$ in
(\ref{F2012(a)}). Hence, assume $E^{\leq }=E^{=}=\mathbf{0}$ in (\ref%
{F2012(a)}). Then, clearly, the ``if and only if'' stipulation of (\ref%
{F2012(a)}) cannot be satisfied in general, since 
\begin{equation*}
(\exists y:\mathbf{0}\cdot x+F^{\leq }y\leq g^{\leq },\text{ }\mathbf{0}%
\cdot x+F^{=}y\leq g^{=})\text{ cannot imply (}Ax\leq b)\text{ in general.}
\end{equation*}

Hence, Theorem 3 of Fiorini \textit{et al}. (2012) is not valid for relating 
$U$ and $X$ in Definitions \ref{Extended_Polytope_Dfn}-\ref%
{Extended_Polytope_Dfn3}, when $G=\mathbf{0}$.\ \ \medskip\ 

Hence, the developments in Fiorini \textit{et al.} (2012) are not valid for
relating $U$ and $X$ in Definitions \ref{Extended_Polytope_Dfn}-\ref%
{Extended_Polytope_Dfn3} when $G=\mathbf{0}$ in those definitions. \ \ \
\medskip
\end{proof}

\subsection{The case of the Minimum Spanning Tree Problem: \newline
Redundancy matters when ``$G=\mathbf{0}$'' \label{Min_Span_Tree_SubSection}}

The consideration of ``$G=\mathbf{0}$'' we have introduced in this paper is
an important one because, as we have shown, it refines the notion of EFs by
separating the case in which the notion is meaningful from the case in which
the notion is degenerate and ambiguous. The degeneracy (when ``$G=\mathbf{0}$%
'') stems from the fact that every polytope is potentially an EF of every
other polytope, as we have shown in section \ref{Non-applicability_Section}
of this paper. The ambiguity stems from the fact that one would reach
contradicting conclusions \ as to what is/is not an EF of a given polytope,
depending on what we do with the redundant constraints and variables which
are introduced. This is illustrated in the following example.

\begin{example}
Refer back to the numerical example in Part $(iii)$ of the proof of Theorem %
\ref{Thm1}. An example of an inequality description of $X$ in that numerical
example is: $\overline{X}=\{x\in \mathbb{R}_{\nless
}^{3}:x_{1}-x_{2}+x_{3}=6;$ $x_{1}+x_{2}\geq 3;$ $x_{1}+x_{3}\leq 7;$ $%
x_{2}+x_{3}\geq 6;$ $\ x_{1}\leq 2\}$. (It easy to verify that the feasible
set of $\overline{X}$ is indeed $\{(2,1,5)^{T}\}$.) Let $U^{\prime }$ denote 
$\overline{U}$ augmented with the constraints of the linear map, $x-Aw=0$.
Clearly $U^{\prime }$ \textbf{does} project to $X$ under the standard
definition (Definition \ref{Extended_Polytope_Dfn}), whereas $\overline{U}$
does not. Hence, the answer to the question of whether or not $\overline{U}$
is an \textit{extended formulation} of $X$ under the standard definition
depends on what we do with the redundant constraints, $x-Aw=0$. If these
constraints are added to $\overline{U},$ then $\overline{U}$ becomes $%
U^{\prime }$, and the answer is ``Yes.'' If these constraints are left out,
the answer is ``No.'' Hence, the \textit{extended formulations} relation
which is established between $\overline{U}$ and $X$ under Definition \ref%
{Extended_Polytope_Dfn2} is ambiguous (in addition to being degenerate, as
we have shown in Theorem \ref{Thm1}). \ \ $\square $
\end{example}

A well-researched case in point for the discussions above is that of the
Minimum Spanning Tree Problem (MSTP). Without the refinement brought by the
distinction we make between the cases of $G=\mathbf{0}$ and $G\neq \mathbf{0}
$ in Definitions \ref{Extended_Polytope_Dfn}-\ref{Extended_Polytope_Dfn3},
the case of the MSTP would mean that it is possible to \textit{extend} an
exponential-sized model into a polynomial-sized one by (simply) adding
redundant variables and constraints to it (i.e., \textit{augmenting} it),
which is a clearly-unreasonable/out-of-the-question proposition. To see
this, assume (as is normally done in EF work) that the addition of redundant
constraints and variables does not matter as far EF relationships are
concerned. Since the constraints of Edmonds' model (Edmonds (1970)) are
redundant for the model of Martin (1991), one could \textit{augment}
Martin's formulation with these constraints. The resulting model would still
be considered a polynomial-sized one. But note that this particular \textit{%
augmentation }of Martin's model would also be an \textit{augmentation} of
Edmonds' model. Hence, the conclusion would be that Edmonds'
exponential-sized model has been \textit{augmented} into a polynomial-sized
one, which is an impossibility, since one cannot reduce the number of facets
of a given polytope by simply adding redundant constraints to the inequality
description of that polytope. The distinction we are bringing to attention
in this paper explains the paradox, as further detailed below.

\begin{example}
\label{EF_MST_Example}We show that Martin's polynomial-sized LP model of the
MSTP is not an EF (in a non-degenerate, meaningful sense) of Edmonds's
exponential LP model of the MSTP, by showing that there exists a
reformulation of Martin's model which does not require the variables of
Edmonds' model (which is essentially the equivalent of having $G=\mathbf{0}$
in the description of $U$ in Definitions \ref{Extended_Polytope_Dfn}-\ref%
{Extended_Polytope_Dfn3}).
\end{example}

\begin{itemize}
\item \textbf{Using the notation in Martin(1991), i.e.:}

\begin{itemize}
\item $N:=\{1,\ldots ,n\}$ \ \ (Set of vertices);

\item $E:$ \ Set of edges;

\item $\forall S\subseteq N,$ $\gamma (S):$ Set of edges with both ends in $%
S $.\medskip
\end{itemize}

\item \textbf{Exponential-sized/``sub-tour elimination'' LP formulation
(Edmonds (1970)):}

$(P)$:

$\left| 
\begin{tabular}{ll}
$\text{Minimize:}$ & $\sum\limits_{e\in E}c_{e}x_{e}$ \\ 
&  \\ 
$\text{Subject To:}$ & $\sum\limits_{e\in E}x_{e}=n-1;$ \\ 
&  \\ 
& $\sum\limits_{e\in \gamma (S)}x_{e}\leq \left| S\right| -1;$ \ $\ S\subset
E$\ $;$ \\ 
&  \\ 
& $x_{e}\geq 0$ \ for all $e\in E.$%
\end{tabular}%
\text{ \ }\right. $\medskip

\item \textbf{Polynomial-sized LP reformulation (Martin (1991)):}

$(Q)$:

\ $\left| 
\begin{tabular}{ll}
$\text{Minimize:}$ & $\sum\limits_{e\in E}c_{e}x_{e}$ \\ 
&  \\ 
$\text{Subject To:}$ & $\sum\limits_{e\in E}x_{e}=n-1;$ \\ 
&  \\ 
& $z_{k,i,j}+z_{k,j,i}=x_{e};$ \ \ $k=1,\ldots ,n;$ \ $e\in \gamma
(\{i,j\}); $ \\ 
&  \\ 
& $\sum\limits_{s>i}z_{k,i,s}+\sum\limits_{h<i}z_{k,i,h}\leq 1;$ $\ \
k=1,\ldots ,n;$ $\ \ i\neq k;$ \\ 
&  \\ 
& $\sum\limits_{s>k}z_{k,k,s}+\sum\limits_{h<k}z_{k,k,h}\leq 0;$ \ $%
k=1,\ldots ,n;$ \\ 
&  \\ 
& $x_{e}\geq 0$ \ for all $e\in E$; \ \ $z_{k,i,j}\geq 0$ \ for all $k,$ $i,$
$j.$%
\end{tabular}%
\text{ \ }\right. \medskip $\ 

\item \textbf{Re-statement of Martin's LP model:}

For each $e\in E:$

- \ Denote the ends of $e$ as $i_{e}$ and $j_{e},$respectively;

- Fix an arbitrary node, $r_{e}$, which is not incident on $e$

\ \ (i.e., $r_{e}$ is such that it is not an end of $e$).

\item Then, one can verify that $Q$ is equivalent to:\medskip \newline
$(Q$'$)$:\medskip \newline
\begin{tabular}{l}
\ \ \ 
\end{tabular}%
$\left| 
\begin{tabular}{ll}
$\text{Minimize:}$ & $\sum\limits_{e\in
E}c_{e}z_{r_{e},i_{e},j_{e}}+\sum\limits_{e\in E}c_{e}z_{r_{e},j_{e},i_{e}}$
\\ 
&  \\ 
$\text{Subject To:}$ & $\sum\limits_{e\in
E}z_{r_{e},i_{e},j_{e}}+\sum\limits_{e\in E}z_{r_{e},j_{e},i_{e}}=n-1;$ \\ 
&  \\ 
& $%
z_{k,i_{e},j_{e}}+z_{k,j_{e},i_{e}}=z_{r_{e},i_{e},j_{e}}+z_{r_{e},j_{e},i_{e}}; 
$ \ \ $k=1,\ldots ,n;$\ $\ \ e\in E;$ \\ 
&  \\ 
& $\sum\limits_{s>i}z_{k,i,s}+\sum\limits_{h<i}z_{k,i,h}\leq 1;$ \ \ $i,$ $%
k=1,\ldots ,n:i\neq k;$ \\ 
&  \\ 
& $\sum\limits_{s>k}z_{k,k,s}+\sum\limits_{h<k}z_{k,k,h}\leq 0;$ \ \ $%
k=1,\ldots ,n;$ \\ 
&  \\ 
& $z_{k,i,j}\geq 0$ \ for all $k,$ $i,$ $j.$%
\end{tabular}%
\text{ \ }\right. $
\end{itemize}

\noindent $\square \medskip $

\begin{claim}
\label{EF_MSTP_Rmk}We claim that the reason EF work relating formulation
sizes does not apply to the case of the MSTP is that although Martin's model
can be made to project to Edmond's model, that projection is
degenerate/non-meaningful in the sense we have described in this paper.
\medskip
\end{claim}

\section{Alternate/Auxiliary Models\label{Auxiliary_Model_Section}}

In this section, we provide some insights into the meaning of the existence
of an affine map establishing a one-to-one correspondence between polytopes
when the sets of descriptive variables are disjoint, as brought to our
attention in private e-mail communications by Kaibel (2013), and Yannakakis
(2013), respectively. The linear map stipulated in Definition \ref%
{Extended_Polytope_Dfn2} is a special case of the affine map. Referring back
to Definitions \ref{Extended_Polytope_Dfn}-\ref{Extended_Polytope_Dfn3}, we
will show in this section that when $G=\mathbf{0}$ in the expression of $U$
and there exists a one-to-one affine mapping of $X$ onto $U$, then $U$ is
simply an alternate model (a ``reformulation'') of $X$ which can be used, in
an ``auxiliary'' way, in order to optimize any linear function of $x$ over $%
X $, without any reference to/knowledge of an inequality description of $X$.

\begin{example}
\label{EF_Insight_Rmk1} \ 

\begin{itemize}
\item Let:

- $\ x\in \mathbb{R}^{p}$ and $w\in \mathbb{R}^{q}$ be disjoint vectors of
variables;

- $\ X:=\{x\in \mathbb{R}^{p}:Ax\leq a\}\ \ \ $(where $A\in \mathbb{R}%
^{m\times p}$, and $a\in \mathbb{R}^{m}$);

- $\ U:=\{w\in \mathbb{R}^{q}:Dw\leq d\}$ $\ \ $(where $D\in \mathbb{R}%
^{n\times q}$, and $d\in \mathbb{R}^{n}$);

- $\ L:=\{(x,w)\in \mathbb{R}^{p+q}:x-Cw=b\}$ \ \ (where $C\in \mathbb{R}%
^{p\times q}$, and $b\in \mathbb{R}^{p}).$

\item Assume that the non-negativity requirements for $x$ and $w$ are
included in the constraints of $X$ and $U$, respectively, and that $L$ is
redundant for $X$ and for $U$.

\item Then, it is easy to see that the optimization problem, \textit{Problem
LP}$_{1}$:\medskip

$\left| 
\begin{tabular}{ll}
$\text{Minimize:}$ & $\alpha ^{T}x$ \\ 
&  \\ 
$\text{Subject To:}$ & $(x,w)\in L;$ $\ w\in U$ \\ 
&  \\ 
\multicolumn{2}{l}{(where $\alpha \in \mathbb{R}^{p}).$}%
\end{tabular}%
\text{ \ }\right. \medskip $

is equivalent to the smaller linear program, \textit{Problem LP}$_{2}$%
:\medskip

$\left| 
\begin{tabular}{ll}
$\text{Minimize:}$ & $\left( \alpha ^{T}C\right) w+\alpha ^{T}b$ \\ 
&  \\ 
$\text{Subject To:}$ & $w\in U$ \\ 
&  \\ 
\multicolumn{2}{l}{(where $\alpha \in \mathbb{R}^{p}).$}%
\end{tabular}%
\text{ \ }\right. \medskip \medskip $

\item Hence, if $L$ is the graph of a one-to-one correspondence between the
points of $X$ and the points of $U$ (see Beachy and Blair (2006, pp.
47-59)), then, the optimization of any linear function of $x$ over $X$ can
be done by first using \textit{Problem LP}$_{\mathit{2}}$ in order to get an
optimal $w,$ and then using Graph $L$ to ``retrieve'' the corresponding $x$.
Note that the second term of the objective function of \textit{Problem LP}$_{%
\mathit{2}}$ can be ignored in the optimization process of \textit{Problem LP%
}$_{\mathit{2}},$ since that term is a constant.\medskip

Hence, if $L$ is derived from knowledge of the vertex description of $X$
only, then this would mean that the inequality description of $X$ is not
involved in the ``two-step'' solution process (of using \textit{Problem LP}$%
_{\mathit{2}}$ and then Graph $L$), but rather, that only the vertex
description of $X$ is involved. $\ \ \square $
\end{itemize}
\end{example}

Hence, when $G=0$, the existence of the linear map, $\pi ,$ stipulated in
Definition \ref{Extended_Polytope_Dfn2} does not imply that $U$ is an EF of $%
X$, but rather that $U$ can be used to solve the optimization problem over $%
X $ without any reference to/knowledge of an inequality description of $X,$
if $\pi $ is not derived from an inequality description of $X$. $\ \ $

\section{Conclusions\label{Conclusions_Section}}

We have shown that \textit{extended formulations} theory aimed at comparing
and/or bounding sizes of inequality descriptions of polytopes are not
applicable when the set of the descriptive variables for those polytopes are
disjoint (i.e., when ``$G=\mathbf{0}$''). We have illiustrated our ideas
using the Fiorini et al. (2011; 2012) developments, and Martin's (1991) LP
formulation of the MSTP, respectively. We have also shown that the ``$G=%
\mathbf{0}$'' consideration we have brought to attention explains the
existing paradox in EF theory (typified by the case of the MSTP), which is
that by simply adding redundant constraints and variables to a model of
exponential size one can obtain a model of polynomial size. We believe these
constitute important, useful steps towards a more complete definition of the
scope of applicability for EF's. \ \ \pagebreak

\end{document}